\documentclass[sigconf]{acmart}

\usepackage{hyperref}
\usepackage{amsmath}
\usepackage{url}

\usepackage{listings}
\usepackage[frozencache,cachedir=.]{minted}
\usepackage{tikz}
\usepackage{tabularx}
\usepackage{paralist}
\usepackage{flushend}
\usepackage{amsthm}
\usepackage[inkscapepath=.]{svg}
\usepackage [
n , 
advantage ,
operators ,
sets ,
adversary ,
landau ,
probability ,
notions ,
logic ,
ff ,
mm,
primitives ,
events ,
complexity ,
asymptotics ,
keys
]{cryptocode}

\newtheorem{theorem}{Theorem}

\usetikzlibrary{automata, positioning, arrows}
\tikzset{
	->, 
	>=stealth, 
	node distance=2.5cm, 
	every state/.style={thick, fill=gray!7}, 
	initial text=$ $, 
    font=\footnotesize
}

\newcommand{\fanplugins}{\texttt{plugins} }

\copyrightyear{2024}
\acmYear{2024}
\setcopyright{acmlicensed}
\acmConference[WPES '24]{Proceedings of the 23rd Workshop on Privacy in the Electronic Society }{October 14--18, 2024}{Salt Lake City, UT, USA}
\acmBooktitle{Proceedings of the 23rd Workshop on Privacy in the Electronic Society (WPES '24), October 14--18, 2024, Salt Lake City, UT, USA}
\acmDOI{10.1145/3689943.3695038}
\acmISBN{979-8-4007-1239-5/24/10}

\begin{document}

\title{Towards Flexible Anonymous Networks}

\author{Florentin Rochet}
\orcid{0000-0001-5275-9308}
\affiliation{%
  \institution{University of Namur}
  \city{Namur}
  \country{Belgium}}
\email{florentin.rochet@unamur.be}

\author{Jules Dejaeghere}
\orcid{0000-0002-4970-3730}
\affiliation{%
  \institution{University of Namur}
  \city{Namur}
  \country{Belgium}}
\email{jules.dejaeghere@unamur.be}

\author{Tariq Elahi}
\orcid{0000-0001-7559-8383}
\affiliation{%
  \institution{University of Edinburgh}
  \city{Edinburgh}
  \country{United Kingdom}}
\email{t.elahi@ed.ac.uk}

\begin{abstract}
  Anonymous Communication designs such as Tor build their security on
  distributed trust over many volunteers running relays in diverse global
  locations. 
  In practice, this distribution leads to a heterogeneous network in which
  many versions of the Tor software co-exist, each with differing sets of
  protocol features. Because of this heterogeneity,
  Tor developers employ \textit{forward-compatible
  protocol design} as a strategy to maintain network extensibility.  This strategy aims to guarantee that 
  different versions of the Tor software interact
  without unrecoverable errors.  In this work, we cast  \textit{protocol
  tolerance} that is enabled by forward-compatible protocol considerations as a
  fundamental security issue.  We argue that, while being beneficial for the
  developers, protocol tolerance has resulted in a number of strong
  attacks against Tor in the past fifteen years.
  
  To address this issue, we propose Flexible Anonymous Network (FAN), a
  new software architecture for volunteer-based distributed networks that
  shifts the dependence away from protocol tolerance without losing the
  ability for developers to ensure the continuous evolution of their
  software.
  We
  \begin{inparaenum}[i)]
  \item instantiate an implementation,
  \item evaluate its overheads and,
  \item experiment with several of FAN's benefits to defend against a severe
  attack still applicable to Tor today.
  \end{inparaenum}
\end{abstract}



\maketitle

\section{Introduction}

The Robustness Principle, as expressed in connection with the standardization of
the TCP/IP protocols~\cite{rfc793}, a cornerstone of the Internet, is stated
as: “Be conservative in what you do, be liberal in what you accept from others”.
This principle, buoyed by the success of TCP/IP, has become accepted good
practice for subsequent protocol design.
Fundamentally, it advises protocol designs to tolerate, without failing,
unexpected messages. From a systems engineering perspective, this law enables
protocols to be \emph{forward compatible} and sustain \emph{extensible} protocol
designs. Protocols that
are forward compatible allow processing messages from future versions
of the same software, where there may be parts that the now older version does
not know how to process and still perform correctly.
The application of the Robustness principle has had a considerable beneficial
effect on implementation, maintenance, and deployment of protocol features and
underlying distributed systems, making compatibility straightforward between
multiple protocol versions.
When the TCP/IP protocols were being designed, and the Robustness
principle was established, it was difficult to foresee the security and
privacy implications and the extent of what the Internet would
eventually become.

Since its initial design in the early 2000s, the Tor routing protocol~\cite{tor-design} also implements the
Robustness principle when
processing data and control information. Indeed,
robustness is crucial in a distributed and volunteer-based network composed of nodes
running many different versions of the Tor routing protocol. Nodes running older
versions of the codebase must tolerate messages with extra or updated data
formats from nodes running the latest version.
Rochet and Pereira~\cite{dropping-pets2018} showed that, as a
\emph{threat vector}, protocol tolerance could be exploited to convey
information between malicious relays that easily
defeat the anonymity enabled by Tor's incremental circuit design. In the same vein, we make the
observation that other
researchers have also unknowingly exploited this threat vector, such as fast and reliable
Onion Service's Guard discovery~\cite{hs-attack06, oakland2013-trawling,
	dropping-pets2018} or efficient
Denial of Service~\cite{sniper14, pointbreak-sec2019}. Furthermore, a
real-world attack sponsored by a
state agency has been observed in the Tor network using the same threat
vector~\cite{blogpost-relayearly, blogpost-didfbi}. All these attack designs
fundamentally exploited some form of protocol tolerance embedded in Tor, which was purposefully
engineered
for forward compatibility reason. Altogether, these developments raise awareness of the role this
principle plays in efficiently and reliably implementing well-known theoretical
attacks, like end-to-end traffic confirmation.

Our contributions are (in order of appearance):
\begin{inparaenum}

\item Inspired by prior work~\cite{dropping-pets2018},
we trace the history of recent significant vulnerabilities in Tor and shed light on
		the inherent tension that robustness introduces between allowing system
    designers enough latitude to create a maintainable
    network and limiting the surface area for the adversary to leverage
    (Section~\ref{sec:case}).

  \item Flexible Anonymous Networking (FAN), a novel methodology that builds
      on a new programming paradigm and execution model to address the risks inherent to protocol
      robustness. FAN's unique properties allow authorized developers to globally
      re-program the nodes of the network without operator intervention (Section~\ref{sec:fan}).
      It supports the network's extensibility through
      network-wide re-programming instead of
      built-in protocol tolerance.
      In brief, FAN is a design methodology for the software making up
      the distributed network as a core, stable set of basic features --- A
      network engine, a Cryptography engine, and FAN components. These latter components are a
      JIT compiler able to compile ``plugin'' programs written in a high-level
      language such as \texttt{C} or \texttt{Rust} and compiled to a RISC VM,
      such as \texttt{eBPF}, and finally a Plugin Manager that manages the
      logic of these new capabilities. Moreover, FAN must support an
      OS-independent secure protocol update mechanism to enable unattended
      upgrades of relays issued by the authorities. One efficient update mechanism
      is empirically evaluated in this work.

    \item A Transparency mechanism to address potential governance issues linked to FAN's capabilities
    that enables relay operators to
      transparently audit the plugin issuance process, enabling trust through
      proof of correct behavior (Sections~\ref{sec:case}~\&~\ref{sec:fan}).

    \item An open-source implementation instance of the FAN methodology applied
      to the Tor codebase (Section~\ref{sec:instance_architecture}).

	\item A case study of our FAN instance applied to defending Tor against the
    Dropmark attack~\cite{dropping-pets2018, prop344}, where we implement and
      evaluate a FAN-enabled defense against a powerful attack on Tor that
      is relevant today. We argue that efficiently and provably defending
      against the Dropmark is infeasible without FAN's specific properties
      (Section~\ref{sec:casestudy}).
\end{inparaenum}
\section{Motivation}

\subsection{Threat Model}
Tor assumes a local adversary who can observe some fraction of the network traffic, operate or 
compromise some fractions of 
relays, and interact with clients and other relays by injecting, modifying,
dropping, or delaying traffic.  Adversaries' goals may vary, from disrupting or denying availability to 
linking clients to their destinations. 
The adversary, and the traffic it generates, may also deviate from the Tor protocol.
As we shall discuss below, this deviation is tolerated and
is historically the
result of a desire for Robustness to enable forward-compatible protocols, which has been used by 
adversaries to achieve the mentioned goals.

\subsection{Robustness: A Blessing and a Curse}
\label{sec:case}

We now trace several recent examples of protocol tolerance in action that led to
critical attacks against anonymity or availability of the Tor network, leading
us to argue that it is a systemic problem; a tension between developers who try
to ensure the continuous evolution of their protocol design and adversaries
taking advantage of the ``freedom'' built into the design. Directly linked to
this research, the Tor developers recently (August 2023) publicly acknowledged
the Robustness principle as an engineering issue leading to ``highly-severe''
attacks~\cite{prop344}.

\paragraph*{An Illustration: Extending Tor Circuits}

The current three-hop default for Tor circuits is a trade-off between
performance and anonymity. However, this default is not hard-coded, and
circuits of other lengths are possible. This design choice to parameterize the
default allows it to be easily changed when needed. However, because path
length was not restricted, this allowed researchers to discover a practical
congestion attack~\cite{congestion-longpaths} using infinite hops with the same
relay selected multiple times in a Tor circuit. The Tor project introduced the
new \texttt{RELAY\_EARLY} cell type to mitigate this attack~\cite{prop110},
while maintaining protocol tolerance. Unfortunately, three years later in 2014,
this mitigation enabled a state-level agency to exploit the still present
protocol malleability to correlate traffic with absolute certainty, which
reliably broke Tor's anonymity~\cite{blogpost-relayearly, blogpost-didfbi}. The
Tor protocol was then made more \textit{conservative} (commit 
\texttt{68a2e4ca4baa5}) to prevent tolerance abuse of the new cell type.
	
Moreover, deploying the new \texttt{RELAY\_EARLY} cell type required  backward compatibility with 
	relays that 
  have not yet updated to the latest version of the protocol, and these relays were still
  susceptible to the congestion attack. A tension 
  exists between the desire to quickly close an attack vector while also retaining most of the relays in 
	the network. A grace period allows maximizing the number of updated relays before 
	making the version mandatory and evicting out-of-date relays. In this case, it took three years 
from proposal to implementation (version 0.2.1.3-alpha) and another year to full deployment
  (commit \texttt{9bcb18738747e}), while the congestion attack was still
  possible to perform in the intervening four years.

 To summarize, this historical review highlights that designers and developers desire the 
    elimination of flaws in their protocols without limiting extensibility. From
    above, and other 
    historically similar events~\cite{hs-attack06,  oakland2013-trawling, sniper14,
			dropping-pets2018, pointbreak-sec2019}, we see that extensibility, through protocol tolerance, 
			not only enables new features but also unforeseen vectors for attack. Finally, we note the 
			uptake of corrective measures may take multiple years, leaving the network and its users
      vulnerable in the interim.

\paragraph*{The key challenge}

Developers require the \textit{agility} to enable new features, including security
fixes, to be quickly and widely propagated. A natural, and naive, solution
is \emph{automated updates}. However, in our setting, it is radically more difficult.
Automated updates are usually seen from a client's perspective, where a program, upon execution, first 
checks and fetches the latest version (if available), before running. 
In a distributed system such as Tor, where we are concerned with volunteer-run
\emph{relays}, the situation is different. These
relays are usually left unattended and running on a dedicated machine with
the only task of relaying traffic---as advised by the Tor project. The
update pace of the network depends today on individual relay operators, either
following their OS distribution package update policy, or choosing one of their
own. As Ajmani \textit{et al.} point out in their generic solution for distributed
systems~\cite{ajmani2006modular}, the availability requirement makes the problem difficult. In Tor's
case, it is even more challenging; essentially, a distributed network is required
to migrate altogether to a new version, without disruption of service, and
without privacy degradation during the process (e.g., network partition). Fulfilling these
requirements as well as practical issues such as \begin{inparaenum} \item the
  risk of failing to bootstrap the network, \item the need for per-OS init
  scripts and the cost of maintaining them, and \item the security concerns that
  an auto-update system based on current practices may cause to volunteers,
altogether explain why auto-updates for a distributed system
such as Tor do not exist.\end{inparaenum} 

\subsection{Tor Project's Direction \& Our Vision}

Tor developers produce new versions of Tor and
make both the source and binaries available on their software repositories and
make announcements about the availability of new versions on their official channels.

Relay operators monitor these official channels to keep track of new releases, which they update to at 
their own discretion. On the other hand, Linux
distributions do not monitor these channels. Instead, it is the Tor project, or some party that
wishes to make Tor software available on the distribution, that produces a package (that
adheres to the distribution's rules) and submits it for inclusion. It is to the
distribution's discretion if the package is included or not (and this applies
to updated version of already accepted packages). Often, as is the case with
many popular mainstream Linux distributions, the packaged version of Tor is
many versions behind the official latest release. One critical outcome of the
discretionary nature of deploying/adopting versions, is that  
the current Tor
network consists of relays running several major Tor versions. This fact forces
the Tor developers to adopt techniques to build tolerance within their protocol
design and features, in an attempt to avoid issues over time, especially
since tolerance avoids evicting relays running older versions from locations and volunteers that 
together contribute to
Tor's diversity, which is critical for Tor's security model. This was the situation
up to 2019, but since then, the Tor project decided to slowly move away from
Robustness by enabling strict protocol negotiation, and eventually with the
plan to reject messages that are invalid according to authorized versions. Such a plan
leads to End-of-Life cycles, in which relays
running software versions that are no longer supported are rejected from the
network, even when their contribution is meaningful in terms of bandwidth or
specific security attributes (e.g., location diversity).

As of early 2024, Tor has, however, the downsides of both approaches. Indeed, the
Tor routing protocol is still applying the Robustness principle, which still
enables stealthy attacks such as the Dropmark~\cite{dropping-pets2018}, and
applies protocol negotiation with end-of-life policies, rejecting hundreds of
relays every year~\cite{eol}. By the end of 2024, the Tor project should have
implemented Proposal 349~\cite{prop349} ``Client-Side Command Acceptance
Validation'' which employs state machines to check and verify the various Tor
protocols' state on the client, and tear down any circuit that receives an
invalid message that would not yield a valid state transition. This approach
seems so far to only be considered for clients, and Robustness is maintained on
relays

We believe protocol negotiation that rejects relays is particularly
unfit for a volunteer-based system, with operators in many cases not directly
responsible for running an outdated relay (e.g., the default OS distribution repository providing
an outdated package).
Programming languages and compiler tools have sufficiently improved to revisit
how we build and distribute a distributed system, and that we can obtain the
benefits of the Robustness principle without its security downsides, as well
as the benefits of protocol negotiation without rejecting anyone. Our technical
proposal to successfully achieve these two seemingly contradicting properties is our main
contribution. We furthermore show in our case study
(Section~\ref{sec:casestudy}) that with these
two properties in hand, we can solve and prevent one
of the most efficient attack vectors against Tor's
anonymity.

Our approach involves a new kind of protocol negotiation design, where instead
of only negotiating the version to use and rejecting non-compliant nodes, we additionally offer a method 
to 
\emph{re-program} parts of the software of non-compliant nodes,
hence propagating new code as part of the negotiation protocol within circuit
creation.  The challenge is to enable this without disruption, and without
partitioning the network. Moreover, this process should be transparent and
auditable to Tor clients and Tor operators, and should not depend on a
particular distribution or operating system to maintain diversity and the
ability to run the Tor program on any operating system. The system should also
have safety guarantees in case of a bootstrapping error. We also require minimization of the  bandwidth 
overhead that such a protocol could create.

Unfortunately, this list of requirements is so restrictive under the current methods that we need to 
re-think how
we build and distribute software. We believe that such a
challenging task can be solved by having the software of a distributed system
such as Tor include a Just-in-Time compiler (JIT), which would be leveraged
by a specific and secure version negotiation protocol to propagate and update
portion(s) of the binary to change its features and effectively ``update'' the
network. We propose a method for a distributed system that provides the ability to re-write and
extend itself while it is running, enabling new connections to switch to the
new code, while the existing connections finish operating over the old code. In
this paper, we use the eBPF bytecode compilation target to propagate the
updates, and an eBPF-to-x86/arm JIT compiler to locally re-write the binary once the
update is received without restarting the service and causing any disruption
and network partition.

\subsection{Background on eBPF, LLVM and JIT Compilation}

Our solution reconciles protocol tolerance and its potential risks by
developing a novel architecture to secure distributed networks. This
architecture takes advantage of recent progress in programming language and
compilers, such as simplification of compiling \texttt{C} or
\texttt{Rust} code to machine-independent RISC instruction sets. 
eBPF~\cite{eBPF-site} (also called BPF) is an example and mainly known for being used in the
Linux Kernel to improve the kernel's flexibility to specific needs.  

However, eBPF is a \textit{general}-purpose RISC instruction set and independent of
the Linux Kernel. One can write
programs in a subset of \texttt{C} and compile to the BPF architecture using the
LLVM~\cite{LLVM_CG04,LLVM-site} clang frontend and LLVM backend to obtain a BPF
program.
This BPF program contains bytecode instructions matching the BPF
Virtual Machine, and can be either interpreted or compiled once more to native
opcodes. This process of compiling or translating from the BPF bytecode to
native opcodes is called JIT (Just-In-Time) compilation, and it elegantly provides code
portability and efficiency. In our solution, eBPF programs execute within the same process as
their caller and adhere to the host OS's security model. To further isolate the
eBPF program from its caller, we also enable sandboxing of the memory accesses of
the eBPF programs, which we describe in Section~\ref{sec:instance_architecture}.

\section{FAN Architecture}\label{sec:fan}

Based on new software execution models enabled by technologies such as LLVM,
Flexible Anonymous Networks (FAN) is a design for deploying and maintaining an
anonymous network that can seamlessly change its
behavior---through the addition and/or removal of protocol features and
components---without having to restart relays or
disturb users' connections. 


\subsection{FAN Design Features} 
\label{sec:fan_features}

We have argued in Section~\ref{sec:case} that the current method for designing, iterating, and
deploying
features central to a security property (anonymity in the case of Tor) is
lacking the necessary control to handle unforeseen issues. 

We now take a step back and consider a more general view and aim at establishing a set of design features 
that would allow a distributed system
to grow beyond what an auto-update mechanism for classical software can achieve,
discuss potential drawbacks and illustrate some of the benefits through a case
study.

\paragraph*{Built-in Extensibility.}
\label{feature:2}
FAN enables high-levels of \textbf{expressiveness} by exposing internal
data structures, memory, and functions to the plugins, leading to similar expressive
abilities as high-level languages. FAN
supports \textbf{ease of deployment} by offering fine-grained control of the
ability to remotely re-program components of the network by sending plugins to them. In addition, this
novel capability offers \textbf{speedy propagation} of
patches or
feature updates throughout the network, as they become available. This
contrasts with the current update lifecycle of distributed
networks such as Tor which can take up to several years to propagate major
changes. Compared to a classical auto-update
mechanism which would swap binaries and restart the program, FAN achieves
better safety and robustness, and unlocks novel properties for Anonymous
Communication: FAN \fanplugins can be \textbf{connection specific}, and
different plugins can be applied at the same code location depending on the
context.  This capability leads to a novel extensibility paradigm that could
potentially allow different applications, under an appropriate threat model and safety
\& fairness considerations,  to make relays behave differently for their own
set of users through connection-specific re-programmed features.

\paragraph*{Security \& Safety.}
\label{feature:3}
FAN-based networks depend on multiple layers of security. For
\textbf{operational security}, the virtualization technology (detailed in
Section~\ref{sec:instance_architecture}) provides the sandboxing and resource management
capabilities that ensure that FAN \fanplugins may not exceed their memory
consumption limit, and provides the necessary control to safely handle other failure
modes of plugins.  The \textbf{secure deployment} of FAN \fanplugins ensures
that only \textit{trustworthy} and \textit{authentic} code may be deployed in a
FAN. We control these aspects through our FAN Plugin Transparency
design, detailed in Section~\ref{sec:fan_transparency}.

\paragraph*{High-performance.}
\label{feature:4}
FANs introduce \textbf{negligible overhead}, in terms of
computation, storage, and latency. The analysis is provided in
Section~\ref{sec:overhead_analysis}.

\paragraph*{Stronger Feature Control.}
Due to the various existing packaging policies in OSes that conflict with a
distributed network deployment cycle (e.g., Tor and
Ubuntu~\cite{debian-tor-conflict}), and due to inattention from some relay
operators, it is difficult to deploy new major features in volunteer-based
systems. Recall from section~\ref{sec:case} that this has historically
created many protocol-level vulnerabilities that are not bugs, but design flaws
caused by a lack of control over the software's deployment.  We expect a FAN binary to be
composed of a few very stable core features and follows a release cycle that
matches operating systems' major release lifecycle. For example, a new core
might be released every 4 years, matching Ubuntu LTS's lifecycle. Plugins should
not be expected to be compatible to more than 2 cores at any given time. Any
update of the core system within its expected lifetime follows regular
packaging as it is done today, and should be limited to quality of life
improvements and vulnerability fixes covering the core code.

\subsection{FAN Plugin Design}
\label{sec:design_architecture}

\begin{figure}[t]
  \centering
  \includesvg[inkscapelatex=false,width=.93\columnwidth]{fan_generic}
  \caption{Flexible Anonymous Network generic architecture. The hooked plugins
    can be global or connection specific. That is, multiple different plugins
    can potentially be called from a same hook. PM. stands for Plugin
    Manager.}
  \label{fig:fan}
\end{figure}

FAN source code, e.g., our FAN Tor-prototype described in
Section~\ref{sec:instance_architecture}, embeds a JIT-compiler and
can link the JIT-compiled plugin
bytecode to specific
hooks in the FAN software. The FAN source code and the JIT compiler
are
compiled into the FAN binary. When the execution of this binary comes to
a hook, it checks whether a plugin may be executed, otherwise, we run the
default code for that hook. Each hook defines all the information required to
attach a plugin to
that hook in the FAN binary. Depending on the governance model for the
FAN, this information includes
resource budgets (CPUs cycles, memory allocation
size, etc.), external access controls (OS-like file-level authorization,
network
authorization), internal access controls (i.e., the internal FAN data
and structure(s) that a plugin may access, modify, extend, or internal
functions that may be called) and required cryptographic authorizations for
bytecode verification (described in Section~\ref{sec:fan_transparency}).

FAN plugin features are written in a high-level language, e.g., C or Rust, and
compiled independently of the FAN binary, targeting a bytecode
representation. This bytecode targets some Virtual Machine and is not
executable right away (Figure~\ref{fig:fan}(1)).
The plugin has associated meta-information that specifies the hook location
in the FAN binary where the bytecode may be executed
and the capabilities it requires
(Figure~\ref{fig:fan}(2)). A
FAN plugin may attach to one or more hooks with a shared context between the
hooked code.  The meta-information provides this context, which is
instantiated
at loading time (e.g., sandboxed memory). Indeed, the FAN plugin may be
composed of several bytecode files that may be loaded for the same context but
hooked in different locations in the FAN binary (Figure~\ref{fig:fan}(3a)).

The first time the FAN binary loads a plugin (Figure~\ref{fig:fan}(3b)), the plugin
bytecode is
JIT-compiled to the
appropriate machine-level code and then attached to the specified hook. The
same process executing
the FAN binary may then execute this attached code, as often as
needed. Based on
the governance and threat models, specific security measures, like sandboxed
execution (Figure~\ref{fig:fan}(3c)) and restricted access to resources, take effect while executing the plugin
code.

\subsection{FAN Transparency Design}
\label{sec:fan_transparency}

To alleviate trust tensions between developers and relay operators (and users) we propose our Plugin
Transparency design that is adapted from the design in
PQUIC~\cite{de2019pluginizing}, itself inspired from
CONIKS~\cite{melara2015coniks} used for key management. Our design provides relay operators (and 
users) the ability to verify
transparently, or audit after the fact, the plugin issuance process, which
provides public cryptographic evidence of the central authority (i.e. developers) and independent
log's misbehavior.  Operators can observe when a plugin is issued and when it
will be pushed to their node(s). If desired, they can also protest, which is
globally transparent. Protests serve as a community signal; while it cannot prevent
a plugin to eventually be pushed to the network, it can be a tool for convincing
the central authority to withdraw a plugin. The design discretizes system time into epochs,
denoted by monotonically increasing integers $\in \mathbb{Z}_0^+$.

The Plugin Transparency design separates the stakeholders into three sets:
The FAN/plugin developers and relay authorities (if any), network relays and
clients, and \textbf{independent} entities hosting the FAN Transparency Log
(FTL) (a Merkle Tree \emph{variant}~\footnote{~We call it Name-Structured Merkle
  Tree List.  Proving that an element is present or not has the same complexity
  ($\Theta(D + \frac{N}{2^D})$.} with logarithmic complexity for all request
types, and other important properties as detailed next).

The FAN Plugin Transparency design provides several security properties
required to build a trustworthy environment. Similar to the threat modeling
approach of Certificate Transparency~\cite{laurie2021rfc} for the web, or Key
Transparency~\cite{melara2015coniks} for end-to-end encrypted messaging, reputation is gained from
public and cryptographic evidence of correct behavior. Any misbehavior is also
publicly auditable and the environment supports appropriate reactions to keep
the environment reputable. However, we propose design tweaks to
improve performance of proof of absence requests, support potential
third-parties, and capture and map misbehaviors into an existing legal
framework (Trademarks). That is, malicious plugins from a third-party would
have to commit an infringement using a name they do not own.

\begin{inparaenum} 

\item \emph{\textbf{Multiple Namespaces}} This paper discusses a centralized
  governance model in which a single trademark for plugins would be accepted by
  the loading logic on relays. However, we aim our transparency design to be
  trivially extensible to many independent trademarks. Therefore, plugins are arranged in the tree
  following the naming convention: \textit{name := trademark\_name/hook\_name}
  leading to a string-based UID for each plugin.

\item \emph{\textbf{Issuance safety and authenticity}}.
  Multiple FAN developers are cryptographically involved in plugin issuance,
  avoiding a single point of failure. Plugin issuance contains a threshold
  signature and its meta-data including the plugin's name, version, protest and push
  epochs, and status (``deployed'' or ``withdrawn'').  Plugins are issued to
  FTLs. Using TUF~\cite{samuel2010survivable} to provide key compromise mitigation for this specific 
  task would be
  an appropriate choice. A plugin's source code to bytecode ahead-of-time-compilation must be 
  reproducible.

\item \emph{\textbf{Detection of spurious plugins}}. If an FTL maliciously
  injects a plugin within a FAN's namespace, the FAN developers can efficiently
    detect it and cryptographically prove the FTL's misbehavior, with a
    proof of availability. Both the detection and proof run in $\Theta(D +
    \frac{N}{2^D})$ for $N$ plugins and D the tree's depth. The detection needs
    to be done once per epoch, for each existing hook. Note that the detection
    of a malicious entry in our design runs much faster (logarithmic to $N$)
    than typical related works on transparency (linear to $N$). See
    Appendix~\ref{app:sec_analysis} for more details.

  \item \emph{\textbf{Secure and human-meaningful plugin names}}. Plugins'
    names are uniquely and cryptographically mapped to the plugin's bytecode,
    which is a main difference to Certificate Transparency or Key Transparency
    which do not enforce any mapping.  That is, the plugin's position in the
    tree depends on its name's hash value prefix (matching the tree's depth).
    The FAN developers periodically check the proof of availability produced by FTLs
    and verify that no other plugins spoofing the same name are present within
    their branch. Compared to other transparent tree structures such as indexed
    trees (e.g., in CONIKS~\cite{melara2015coniks}), our choice of ordering the
    tree from unique plugin names (e.g., fan.project/my\_plugin) would force
    misbehaving entities to commit to a potential trademark infringement,
    forcing the attacker to use the same namespace, for which a legal framework
    already exists. Different namespaces may collide to the same prefix, hence
    the same position in the tree. The tree's leaf holds a list data structure,
    and the expected number of namespaces per leaf can be directly controlled
    from the size of the binary tree. Having potentially multiple namespaces
    per leaf does not conflict with the security property, but adds a ``load
    factor'' to the proofs' complexity (the $N/2^D$ term), which can be
    engineered to be $\approx 1$ under a regular uniform hashing output
    assumption. Thanks to these design constraints, a proof of absence also run in
    $\Theta(D + \frac{N}{2^D})$ instead of the usual $O(N)$.

  \item \emph{\textbf{Non-equivocation}}. The FTL cannot equivocate plugin
    availability to members of the network (e.g., serving a proof of
    availability to one
    relay and a proof of absence to another) without collusion with the FAN
    developers. In the presence of a single honest FTL, the collusion
    behavior will
		 eventually be detected with cryptographic evidence. (e.g., when the
		 Signed Tree Roots are compared).
 \end{inparaenum}

\paragraph*{Functional Overview}
FAN developers have three sets of actions: The first is to issue and withdraw
new and old plugins, respectively. To issue a plugin, a name and its attached plugin are sent to
the FTLs with meta-information containing a valid threshold signature and a
\textit{protest epoch $E_{protest}$} and \textit{push epoch $E_{push}$}, 
where $E_{protest} \leq E_{push}$. FAN developers can choose to set the protest epoch 
to the future, allowing relay operators the opportunity to audit the plugin 
and potentially officially record a protest. For transparency to be effective, it is important to offer a
deterministic build system to allow reproducibility of the signed bytecode. To withdraw a plugin, only
the name and push epoch is set.
If the protest epoch passes and FAN developers do not withdraw the plugin, the
push epoch defines the plugin's availability in the ACN. In either case, the
plugin will be included in the FTL's Tree at the next epoch. A withdrawn plugin
is still included in the Tree but marked as ``withdrawn'' by appending the
signed withdrawal order to the plugin's meta-info within the Tree once the push
epoch has passed. Second, FAN developers may ask the FTL for a \emph{proof of
availability} for a given plugin, which they can verify to detect the 
inclusion of
spurious plugins (later defined). This proof authenticates the plugin's inclusion
in the FTL's Tree.  The third action is to broadcast cryptographic evidence 
of
the FTL's misbehavior (e.g., cryptographic evidence of a spurious plugin). 

The FTLs have three actions. The first one is building
the Tree for each epoch. The Tree includes all plugins
available to the network at the current epoch. The second action is to
generate proofs and respond to
\emph{Proof of availability} or \emph{Proof of absence} requests for a
given plugin name. The third action
is to collect relays' signed protests against a given plugin received until 
the protest epoch. The signed
protests are appended to the plugin's meta-information.

The relay operators' actions are the following. First, they can verify, 
fetch or share
a \emph{Proof of availability} or \emph{Proof of absence} from an FTL with
regard to the epoch values of a given issued plugin.  Second, relay
operators can optionally interact with the FTLs before the protest epoch
elapses to formally protest against an upcoming
plugin(s) (using their relay's identity key). 

These protests are informational (but globally visible), and FAN developers 
can then decide to withdraw
the plugin or not. If not, then the plugin is pushed to the whole network in 
the 
push epoch, as defined in the plugin issuance order. Relay operators can use 
this information to decide whether to continue participating in the ACN or not.

 \paragraph*{Plugin Governance}
 A question remains about how to set a realistic workflow to populate plugins from the
 central authority, and manage them in interaction with the community. We would suggest a periodic
 workflow anchored to a periodic event. For example, it could be meaningful that
 the Plugin Transparency design chooses an epoch of 1 week, making release of
 new plugins effectively a weekly process, say every Monday Noon UTC. The push
 epoch can be set once a week as well, and the time difference with the plugin
 release would define the protest window.  We believe that anchoring the release
 of plugins to a periodic time event would be human-centered and ease
 community interaction. Once a plugin is released, it would be propagated
 through a negotiation phase during circuit constructions, when nodes
 connect to each others to build a circuit. We evaluate this method in
 Section~\ref{subsec:propagation}.

\section{Design Implementation}
\label{sec:instance_architecture}

We implement a prototype FAN instance using Tor's codebase forked from Tor
version 0.4.5.7. As a virtual machine abstraction, we use eBPF~\cite{eBPF-site}, which is also used in the 
Linux Kernel~\cite{kernel-bpf}. That is, eBPF is the
chosen bytecode compilation target for this prototype depicted in
Figure~\ref{fig:fan} and offers code portability. In this research work, we
wrote the extended features in
\texttt{C} and compile them to eBPF using LLVM's clang compiler that supports
compiling from a subset of the \texttt{C} language towards eBPF. This capability
leveraged from the \texttt{LLVM} project and integrated in our Tor instance
offers the required \textit{expressiveness} to add, modify or replace 
existing protocols
and features within Tor. Plugins can include headers from the Tor source code,
indirectly call existing functions within the Tor source code or access fields and
variables through the plugin manager interface.

As Figure~\ref{fig:fan} shows, to execute plugins, we first need to compile the
eBPF bytecode to native machine code. This capability offers \textit{high
  performance} to our FAN design (see Section~\ref{sec:poc} for details). We
add a modified version of uBPF~\cite{ubpf}, a JIT compiler of eBPF bytecode
for x86 and arm64 targets. In the original Tor codebase we implement  a \emph{plugin
  manager} module. It can load, compile, and link within the Tor codebase and
execute eBPF bytecodes with a shared context, using a shared area of memory for all
the bytecode files belonging to the same namespace (i.e., each plugin has a
namespace). Each of the bytecode files defines what we call an \textit{entry
  point} for the plugin, i.e., a function that can be \textit{hooked} and
called from within the Tor binary similar to a \texttt{main()} function. When a
plugin is loaded, we create a new mapping in the virtual address space of the host, which is intended to hold the
plugin's code. We change this new memory region's flags to give read and executable
rights for the host (using the syscall \texttt{mprotect()}). At
execution, when the code meets an entry point, it jumps to the host's paged
memory linked to the entry point, and continues its execution with the code that
was stored there when the loader JIT-compiled the plugin from eBPF to machine
code. Eventually, the entry point's function return instruction jumps back to its
caller as regular x86/arm64 would do.

Each plugin has its addressable memory sandboxed within the host's (i.e. FAN) process. The memory allocated within
plugins belongs to a dedicated continuous buffer allocated by the FAN process
when JIT-compiling it. During this process, the compiler adds instructions to
check boundaries for the dereferenced memory, and cleanly exits the plugin's
execution if it tries to access memory outside its allocated range. The
plugin receives a virtual address space, and the host (the FAN process) records
an unsigned offset to translate the plugin's address space back to the process' address
space. When the plugin dereferences a pointer, the JIT-compiled instruction
checks if the memory address translated back to the host's memory address space
and size requested are correct (i.e., below the highest address within the dedicated continuous buffer).
This design also means that plugins cannot directly access structures and
memory allocated by the FAN process, since they would be outside their
allocated range.  To access them, our FAN implementation supports indirect
getter and setter functions for various protocol-related data structures. This
method also allows the host to expose what it judges necessary to plugins,
while also supporting extensibility (i.e., the indirect access can also get
extended, and data structures can also be extended). This design makes the
distributed network robust to faulty plugins, and hardens against the exploitation of
potential memory vulnerabilities within the plugins themselves.

The plugin manager module of FAN uses human-readable instructions to load the
entry points. The file is parsed according to the plugin manager's
specification, and the plugin is then hooked as indicated within the
\texttt{.plugin} file if the core protocol supports the selected hook.

Our prototype integration (Plugin Manager) and plugins contain
several thousand lines of \texttt{C}. The modified version of uBPF contains
several tens of thousands lines of \texttt{C}. While we merely modified a
fraction of these to fit our purposes, such code would need to be maintained
for a real integration. FAN's architecture allows independence
between the design of the anonymous distributed network, and the design of the
virtual machine linked to each node. Therefore, 
in the future, recent initiatives such as the
\textit{Wasmer}~\cite{wasmer} project or \textit{Wasmtime}~\cite{wasmtime}
project would be good fits for providing such a security-wise critical abstraction,
in the same way some projects have been fit to write cryptography independently
of any system, such as OpenSSL, which is the current embedded cryptography
engine in Tor. Our code is available online:
\url{https://github.com/flexanon/tor_ebpf}.

\subsection{Engineering Hooks' location}

Within the core of the distributed software's source code, there is no technical restriction on the location 
of hooks. However, we recommend adding hooks for protocol
operations where Robustness was initially engineered, to essentially replace
forward compatibility built from the Robustness principle which currently
benefits both the developers (having many versions deployed without
incompatibility) and the adversary (exploiting protocol tolerance). Forward compatibility would then be 
built from hooks and Robustness would be gained from the
ability to hook code to deal with otherwise previously unknown message. As a result, by
default, any unknown message that cannot be handled by an existing plugin
would be invalid, rejected and any required cleanup would need to take
place (e.g., circuit tear-down), effectively ensuring no tolerance policy for
unknown messages, but still retaining the ability to add/modify or delete a
protocol message and its code processing logic globally to the network.

In our prototype implementation, we replace with hooks the existing Robustness within the Relay
protocol and within the Circuit padding protocol implemented in
C-Tor~\cite{kadianakis2021pcp}. The inserted hooks are able to handle any new
protocol event, and support extensions to add new capabilities to major pieces of Tor. 
Other parts and subprotocols of the Tor system may also receive hooks, however, in this work we cover 
only
the main locations required for our use cases and proof of concept
demonstrations. An example of a hook written in C-Tor is provided in
Appendix~\ref{cf-code}.

\section{Performance Evaluation}
\label{sec:poc}

\subsection{Methodology}

We focus our performance analysis on three questions: First, we are interested
in \textit{the time it takes to propagate new code over the whole network}. To answer
this question, we use Shadow~\cite{shadow-ndss12} and experiment with new plugin deployments from
different vantage points. Second, once we have received the plugin, we're
interested in \textit{the time it takes to load it, compile it, and be ready for
  execution}. This analysis is performed by recording statistics from the
plugin manager module when it loads plugins. Finally, we are interested in
\textit{the runtime overhead of a typical plugin}. To answer this question, we
build a testbed designed to stress-test the usage of plugins.
Altogether, this analysis should demonstrate the feasibility of our solution.

\subsection{Time to propagate within the full network}
\label{subsec:propagation}

We evaluate how long it takes for a new plugin to be distributed across the
network.  To this end, we created a lightweight negotiation and exchange protocol which
negotiates and propagates the latest code during the circuit construction
phase, as a simple extension to Tor's circuit handshake. Therefore, nodes
synchronize on the latest version that is used directly in the handshake phase, and
update each other if necessary---independent of the OS and hardware used on
each relay. When the circuit is built, every node is guaranteed to run the same
(latest) version on the circuit.

We set up a Shadow simulation with a 10\% scaled-down
version of the Tor network as of April 2021 (660 relays) to carry out the
measurements. Five plugins were deployed during the simulation (see Table~
\ref{table:overhead}), which are later also used in the case study in Section~\ref{sec:casestudy}.
We measure the time between the introduction of the plugin in the network and the reception 
by the relays. 
Two scenarios are simulated: seeding\footnote{~Seeding: making the plugin first available
  on those relays and relying on the protocol to propagate them to the other
  relays. Among the thousands of operators, we may assume that some of them are
  actively tracking latest releases and pull the last version when available,
  seeding their node in the process.}
the new plugins from the three authorities in the simulation 
and seeding the plugins from the three fastest relays in the simulation.
The measurements are shown on the left of Figure~\ref{fig:update-time}.

From the left graph in Figure~\ref{fig:update-time}, we observe a fast and steady propagation
of the new plugin across the network. The propagation directly depends on the
overall client activity, as exchange happens during circuit creation. Fast
relays have a higher selection probability for circuits, hence seeding these will provide the fastest 
expected propagation.

From our experimental measurements, it takes 12.9 seconds for a new plugin to be distributed
to 90\% of the relays and after 72.8 seconds, more than 99\% of the relays received the new
plugin (in the worst case simulation). During the first few seconds of distribution, the number of relays
receiving the plugin grows exponentially.

We compare the time needed for our plugin distribution method with the time needed for new Tor
versions to be adopted, depicted in the right graph in Figure~\ref{fig:update-time},
based on metrics from the Tor Project~\cite{tor-metrics}. It is clear that it takes time for independent 
volunteers to adopt new versions and this forces Tor to employ forward compatibility to ensure tolerance 
between the various behaviors across co-existing version.
These observations are not specific to Tor; any distributed network has similar software distribution 
issues.  Note, Tor version 0.4.5 is a long-term support
version of Tor, with a lifetime of 2 years (104 weeks)~\cite{tor_coretorreleases_2023}.
According to the metrics, it took nearly half the lifetime of the version to
reach 90\% of the relays.

\begin{figure} \centering
   \begin{minipage}{.49\linewidth}
     \includegraphics[width=\linewidth]{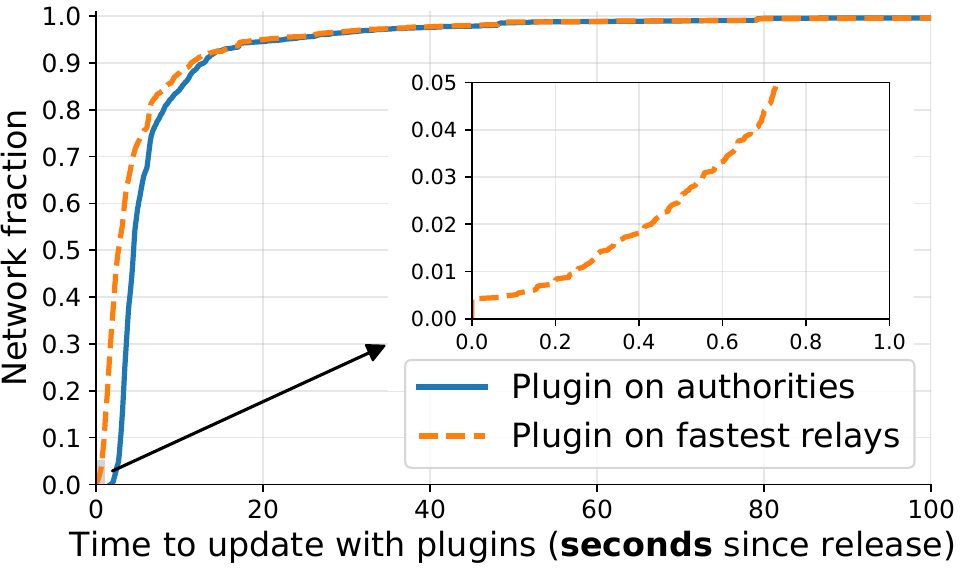}
   \end{minipage}
   \begin{minipage}{.49\linewidth}
     \includegraphics[width=\linewidth]{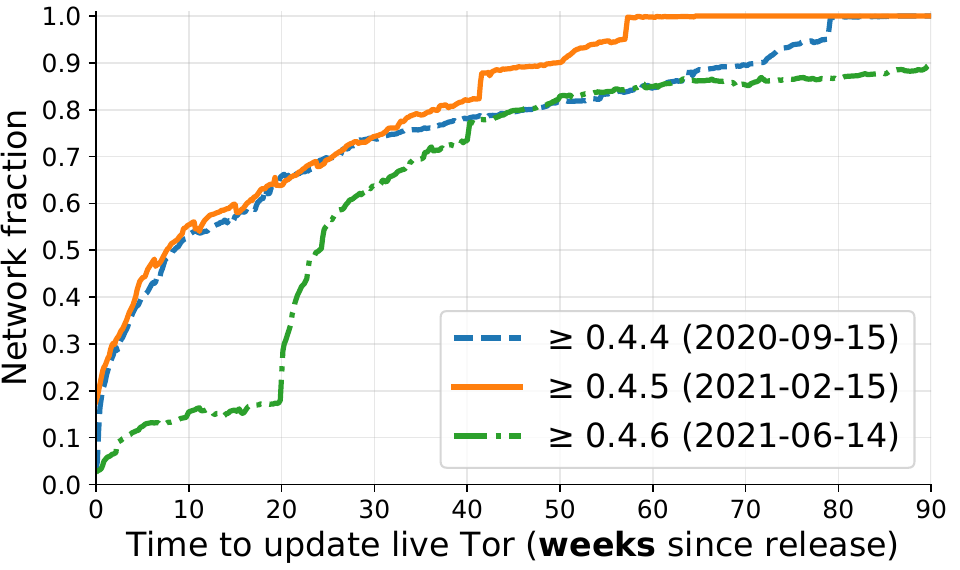}
   \end{minipage}
   \caption{(Left) Shows plugin propagation within a simulated network using
     the FAN architecture (automated and unattended ``push'' of new code). (Right) Shows Tor's situation based on historical
     metrics (a fraction of volunteers perform automated ``pull'' of new code). Note the units.}
   \label{fig:update-time}
\end{figure}

\subsection{Overhead Analysis}
\label{sec:overhead_analysis}

\subsubsection{Loading Plugins}

We evaluate how much CPU time it takes for typical plugins to be ready to
execute. It includes acquiring the bytecode, the JIT compilation CPU time, and
the plugin environment setup in our modified Tor version. Acquiring plugins
means loading them from disk. Additionally, to evaluate the plugin's authenticity, e.g. with Plugin 
Transparency, requires applying a hashing
function $H$ several times to verify the proof of availability received with
the plugin, which consumes on average a few $\mu s$ of CPU time.
A plugin may be composed of several
entry points, each contained in its own compiled bytecode file. Our framework loads
each of them independently and links them at the application-level context, such
that a plugin may be composed of multiple entry points and may hold a shared
sandboxed memory between each loaded entry point.

Table~\ref{table:overhead} gives an overview of different plugin costs. We
implemented several plugins to test the Tor FAN architecture,
including ``system-wide'' plugins that modify or replace existing
code, or connection-specific plugins that can be attached to a
given Tor circuit, and replace or modify a part of the code only
for that circuit. Section~\ref{sec:casestudy} explores these
capabilities in a case study.

From Table~\ref{table:overhead} (i.e., typical plugins are loaded in less than $1ms$), we can
conclude that dynamically loading new (uncached) protocol features is feasible and does not
incur significant latency, which we believe opens the door to different
governance policies. For example, dynamically sending a plugin within a Tor circuit
would be feasible. That is, upon circuit creation, a Tor client can now negotiate
features that relays have to support on this circuit. These features may be sent
and loaded as part of the circuit creation procedure without incurring
significant latency. 

\begin{table}
  \begin{scriptsize}
  \centering
\begin{tabular}{c c c c c}
  Plugin & \#LoC & \#hooks & Bytecode size & Loading Time  \\
  \hline
  \hline
  hello\_world & 10 & 1 & 1,048B & 204 $\mu$s \\
  sendme\_1 & 25 & 1 & 1,672B & 202 $\mu$s \\
  sendme\_2 & 67 & 2 & 3,840B & 385 $\mu$s \\
  sendme\_3 & 82 & 3 & 5,184B & 486 $\mu$s \\
  dropmark\_def & 652 & 4 & 13,488B & 664 $\mu$s \\
  dropmark\_def\_uncons & 675 & 5 & 14,616B & 729 $\mu$s \\
  dropmark\_def\_conn\_based & 322 & 4 & 11,736B & 493 $\mu$s \\
\end{tabular}
\end{scriptsize}
\caption{Load time displayed for \emph{uncached} plugins, and size refers to
  the transmission cost.}
\label{table:overhead}
\end{table}

\subsubsection{Fast-path Code Modification}

We now investigate the impact of running the plugin code with respect to the
legacy code within the original binary. Each code
path (sequence of instructions) does not have the same impact on the software, and
the overall Tor network. For example, we may differentiate the ``fast path'' from
a ``control path''. The fast path are the sequences of instructions 
typically executed when performing the most common task within the network. A
control path is the sequence of instructions performed over a specific and
potentially rare event.
In typical usage of relays, the fast path in Tor involves all instructions relaying
cells to the next hop. It is thus desirable that a FAN architecture supports plugins in
the fast path while maintaining high performance.

Figure~\ref{fig:overhead} shows the design of a private testbed to stress-test our FAN instance
implementation and measure the overhead of
several entry points hooked into the fast path of the exit relay. In this
experiment, 40 Tor clients send 20 MB of data on each circuit. Each circuit shares the
same exit relay.  We measure the throughput of the slowest stream in each
run of the experiment. Thus, $500$ iterations of the experiment are executed for each network 
configuration.

As a baseline, using vanilla (i.e. unaltered) Tor we craft the number of parallel streams in this stress-test 
network to push the
exit relay's process to a 100\% CPU consumption.
This method allows us to capture any overhead induced by our plugins as a
throughput reduction. The processor on the machine used for this evaluation is an
\texttt{AMD Ryzen 7 3700X 8-Core}.

Figure~\ref{fig:overhead} displays the results of three experiments comparing our FAN instance to
vanilla Tor. Each of the experiments made with our FAN instance involves a plugin modifying the Tor
flow-control algorithm with an increased amount of entry points hooked to the
fast path. Each entry point modifies an existing function of the flow-control
algorithm called in the fast path. That is, we test three plugins expected to be
increasingly more costly by requiring more entry points to be hooked in the
fast path. In each case, the plugin reproduces the behavior of the legacy code
as a baseline. The code being intercepted and replaced is executed
each time a data cell is processed on end points (clients and exit relays).

From our results, we make two observations. First, there is a small average
throughput reduction when the network uses plugins, as depicted in the three
experiments running the plugins. Second, increasing the amount of code plugin entry points
does not significantly further decrease the throughput.  The main culprit for
the throughput reduction is actually coming from the hook design implementation
within the main Tor binary.  When executing a plugin, we first have to locate
it in a hashtable using the hook's name. Hashing a string is a costly operation
on the fast path. Further engineering could bring more reduction to the overall
FAN implementation cost.  However, even without optimization these results are positive: they
show that our method to link, deploy and use JIT-compiled machine code from a
portable bytecode abstraction offers near-optimal performance over the fast
path in an anonymous network. This is true for code that does not compile to a
dedicated set of hardware instructions (e.g., AES-NI).

\begin{figure} 
	\centering
    \centering
     \begin{minipage}{.49\linewidth}
      \includesvg[inkscapelatex=false,width=\linewidth]{overheads}
      \end{minipage}
     \begin{minipage}{.49\linewidth}
    \includegraphics[width=\columnwidth]{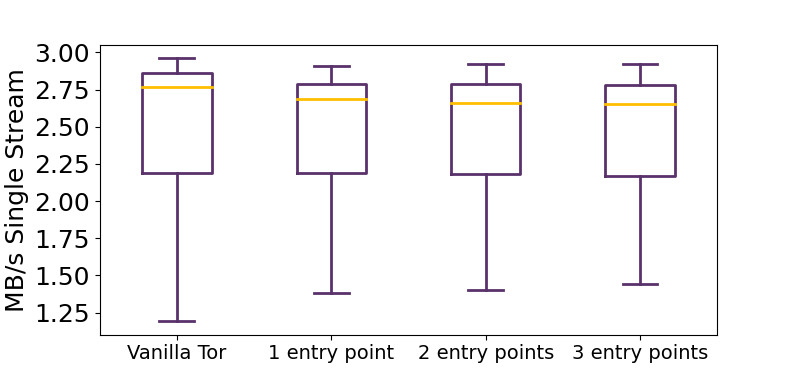}
  \end{minipage}%
  \caption{(left) Measurement setup for FAN overheads estimation over the fast path.
    Each client is configured to send a 20MB stream of data. (right)
    Throughput of the slowest Tor client within the measurements. Each
    boxplot covers 500 runs with a plugin being hooked at 1, 2 or 3 different
    places. We use plugins sendme\_1, sendme\_2 and sendme\_3 from
    Table~\ref{table:overhead}. }
  \label{fig:overhead} 
\end{figure}

\section{FAN Case Study: The Dropmark}
\label{sec:casestudy}

We investigate the use of our Tor FAN instance
(Section~\ref{sec:instance_architecture}) with a case study: defending 
against the Dropmark
attack~\cite{dropping-pets2018} to demonstrate FAN's novelty and its benefits.
The dropmark attack is a traffic confirmation technique exploiting the
Robustness issue discussed in Section~\ref{sec:case}, now considered publicly
as ``highly-severe'' by the Tor project in their recent attempt to categorize
attacks~\cite{prop344}.  We apply FAN's capability to deploy a defense and argue that
this defense is not possible with the current model and only effective over a FAN. We evaluate our
defense using Shadow~\cite{cset12-modeling, shadow-ndss12, rob-usenix2021}
simulations. We design the defense as two cooperative plugins, one plugged-in
``system-wide'' that extends the Tor protocol, and one which is
\textbf{\emph{connection-specific}}, i.e., plugged-in to extend the capability
of a given connection, or in this case, a Tor circuit.  Designing and deploying
this defense illustrates one of FAN's design goals: addressing the weaknesses
stemming from protocol extensibility techniques (i.e., Robustness principle and
Protocol negotiation) with a novel methodology.

\subsection{Case Study Background}

\subsubsection{Dropmark}

The Dropmark attack is an active and reliable traffic confirmation
technique~\cite{dropping-pets2018} having low false positive rate and high
success rate, which can be considered more robust than other existing correlation
strategies~\cite{ndss09-rainbow, ndss11-swirl,
  deepcorr-ccs2018, kohls2018digestor, deepcoffea} due to its unique properties.
In effect, the dropmark attack is a 1-bit communication channel
between a malicious exit relay and a passive observer watching Tor clients, such
as a malicious ISP or a malicious Guard relay. The attack exploits two key
characteristics of network protocols: 1) the existence of periods of silence in
which no data should be expected by a Tor client (which we also call ``protocol
silence''), and 2) the Robustness principle. That is, malicious
exit nodes can inject a message at the precise time of expected protocol
silence. The Tor client would then drop that message. This message does
not trigger circuit breakage due to protocol tolerance and would be observable
to any on-path passive observer. This contrasts to other traffic analysis
methods relying on full data traces and machine learning, which suffer from a number of
limitations~\cite{rimmer2022trace}.
In this case study, we explore how to reduce
the Dropmark attack's reliability efficiently thanks to a plugin-based protocol
extension within Tor's circuit padding framework.

\subsubsection{Circuit Padding}

Tor's
circuit padding~\cite{kadianakis2021pcp} framework is designed to enable Tor clients and relays to 
negotiate
and
``program''~\footnote{~It does it by offering a simple interface to specify
  machine states and transition events.} finite state machines, in which each state matches a padding
behavior according to a specified padding distribution (histograms or
parametrized continuous distributions). Pre-programmed events
within the padding machine framework can capture Tor related events, such
as a circuit opening, to make the machines change
their current state, hence changing their behavior.
Figure~\ref{fig:paddingmachine} shows a padding machine currently being used in
Tor with a state transition based on circuit events. Padding machines may be much more
complex in the future and the framework is designed to foster research and ease
of deployment for solutions tackling complex problems such as website
fingerprinting~\cite{pulls-padding} or traffic confirmation.

\begin{figure}[ht]
  \centering
\begin{tikzpicture}[thick, every node/.style={scale=0.8}]
  \node[state, initial] (q1) {$start$};
  \node[state, right of=q1, xshift=0.6cm] (q2) {setup};
  \node[state, accepting, right of=q2] (q3) {end};
  \node[below of=q1, yshift=1.4cm] (l1) {No Padding};
  \node[below of=q2, yshift=1.4cm] (l3) {Uniform in [0, 1] ms};
  \node[below of=q3, yshift=1.4cm] (l4) {No Padding};
  \draw (q1) edge[above] node {Non-Padding Sent} (q2)
  (q2) edge[bend left, above] node{Padding Sent} (q3)
  (q2) edge[bend right, below] node{State Length is 0} (q3);
\end{tikzpicture}
\caption{Used in Tor relays to make the cell pattern of a
rendezvous circuit construction similar to regular exit circuits.} 
\label{fig:paddingmachine}
\end{figure}

\subsection{The Dropmark Defense}
For this case study, we \textbf{\textit{dynamically extend}} the padding machine framework above 
realizing 
FAN's benefits by addressing one of Tor's
core weaknesses. 
To enable this dynamic extension, we placed generic hooks within Tor's Circuit Padding Protocol that are
designed to support and react to new internal events and new messages sent by
the Tor client. These hooks allow us to redefine the entire original padding
framework, including adding features unforeseen by Tor developers at the time of the initial design. Using 
these
capabilities, we design a \emph{protocol extension}
that supports a padding machine able to mitigate the Dropmark attack. We argue
that while this protocol extension---essentially a conservative and fixed rule---is a necessary condition to 
efficiently address
the Dropmark, it has the potential to negatively impact and conflict with future changes. However, by 
design, FAN 
affords re-programming, both the above and future extensions, thus avoiding these kinds of forward 
compatibility issues. 

\subsubsection{Functional Overview}

Figure~\ref{fig:dropmarkdef} shows the FAN Dropmark defense padding machine
enabled from the two protocol extension plugins ``dropmark\_def'' and ``dropmark\_def\_conn\_based'' from
Table~\ref{table:overhead}, which we deployed and tested in faithful networking
conditions in a Shadow Tor network.
The padding machine supports several events that required an upgrade of the
circuit padding framework. This upgrade is brought through a first plugin
(``dropmark\_def'') to support
sending an application-level event from the Tor client to cause the relay to
transition its padding state to begin padding in the precise moment an
inbound period of silence is expected (known and announced by the client
through the protocol extension). Also, as part of the defense, the plugin
carries another extension enabling a \emph{conservative protocol policy} that prevents the
middle relay (where the machine is plugged-in) to forward any message towards the
client if the circuit is clean (i.e., when no streams have been attached yet),
which defeats any Dropmark attempt before the client decides to actually use
the circuit and attach a stream.  Deploying such an extension on today's Tor is
arguably impossible, as it would hinder the extensibility of the Tor routing
protocol with today's methods (i.e., one does not want to deploy any feature
that could create friction and incompatibilities with any future requirements).
However, with FAN, the protocol can be arbitrarily restrictive since we can
globally re-write any current restrictions to accommodate future
requirements.

On top of the
new protocol events brought by the plugin ``dropmark\_def'', the plugin
``dropmark\_def\_conn\_based'' offers the ability for a given circuit on which
the plugin is plugged-into to make the middle node swap padding cells sent by the
new dropmark machine with regular cells coming from the exit relay. The idea is
to absorb, at the middle node, \textbf{\emph{any}} pattern sent by the exit relay during a
period of silence, and output it towards the guard with the exact same
distribution as the live padding machine. This approach guarantees
thwarting any dropmark attempt while the plugin is active: no matter if the exit is
sending a signal or not, the middle node would output cells according to its state
machine. This plugin performs this task by taking ownership of the circuit
cells transiting through the node, and maintain them within its own forwarding queue. The
cost of this design is extremely light on the network
and have no impact on the user experience, since all of this happens within the
bounds of a protocol silence (i.e., there is no legitimate user activity expected),
as our experiments and analysis show.

This plugin is set up and activated
alongside the \textcolor{red}{Activate} signal from Fig.~\ref{fig:dropmarkdef},
and cleaned from the circuit upon a \textcolor{red}{Be Silent} event which
instructs the padding machine to stop sending padding cells or delaying regular
exit cells to match the padding
distribution.

\begin{figure}[ht]
  \centering
  \begin{tikzpicture} [thick, every node/.style={scale=0.8}]
    \node[state, initial] (q1) {$start$};
    \node[state, right of=q1] (q2) {burst};
    \node[state, right of=q2] (q3) {gap};
    \node[below of=q3, yshift=1.5cm] (l3) {Uniform in [1, gap] ms};
    \node[state, above of=q2] (q4) {silence};
    \node[state, accepting, above of=q3] (q5) {end};
    \draw (q1) edge[below] node{\textcolor{red}{Activate}} (q2)
          (q2) edge[bend left, above] node[xshift=0.6cm, yshift=-0.1cm]{State Length is 0} (q3)
          (q2) edge[bend left] node[rotate=90, yshift=0.2cm]{\textcolor{red}{Be Silent}} (q4)
          (q4) edge[bend left] node[rotate=-90, yshift=0.2cm]{\textcolor{red}{Activate}} (q2)
          (q3) edge[bend left, below] node{Padding Sent} (q2)
          (q4) edge[above] node{Circuit Close} (q5)
          (q3) edge[bend right] node[rotate=-90, yshift=0.2cm]{Circuit Close} (q5)
          (q2) edge[above] node[rotate=45, xshift=0.4cm]{Circuit Close} (q5);
  \end{tikzpicture}
  \caption{Dropmark Defense padding machine negotiated on the middle relay when
    the circuit opens. Events
    ``\textcolor{red}{Activate}'' and ``\textcolor{red}{Be Silent}'' are
    client-triggered events fired through a protocol extension enabled with FAN. State Length is a counter decreasing at each padding cell sent.}
  \label{fig:dropmarkdef}
\end{figure}

In summary, the Dropmark defense padding machine is designed to produce a
Dropmark and absorb any real attacker's attempt. That is, an adversary
observing the network either on the Guard relay or on the wire should observe
the \textbf{\emph{same}} watermark on every Tor circuit regardless of
\textbf{\emph{adaptive}} malicious exit relays injecting traffic on the other end of the
circuit, rendering this traffic confirmation technique unreliable, since the
presence or absence of an adversarial exit exploiting Tor's protocol robustness
features is independent of the observed behavior.

\subsection{Dropmark Defense Analysis}

\subsubsection{Dropmark Attack Replication Results}
We re-implemented the Dropmark attack~\cite{code-dropmark} on a recent Tor version (Tor-0.4.5.7). 
With this updated code, we reproduced the
original results using Shadow 3.0 with a user simulation model
based on real traces collected with PrivCount~\cite{privcount-ccs2016, elahi2014privex}.
Table~\ref{table:dropmark} shows the attack accuracy computed from correct
(True Pos., True Neg.) and
incorrect (False Pos., False Neg.) classifications. 
    The experiment contains $24,800$ simulated Tor users and 202
    Tor relays. All combined, those users connect more than $640,000$ times during
    the experiment to simulated Top Alexa websites.
These results confirm the
original high accuracy of the Dropmark attack over a more faithful network
simulation, thanks to the recent progress in Shadow's simulation
environment~\cite{rob-usenix2021}.

\begin{table}
  \begin{scriptsize}
    \centering
    \begin{tabular}{c c c c c c}
      Attack & TPos & FPos & TNeg & FNeg & Accuracy \\
      \hline
      \hline
      Dropmark & 0.999972 & 0.007713 & 0.99229 &  0.000028 & 0.9961\\
    \end{tabular}
  \end{scriptsize}
\caption{Dropmark attack Shadow Simulation results}
 \label{table:dropmark}
\end{table}

\subsubsection{Empirical Dropmark Defense Results}

\begin{figure}
  \centering
    \includegraphics[width=0.7\linewidth]{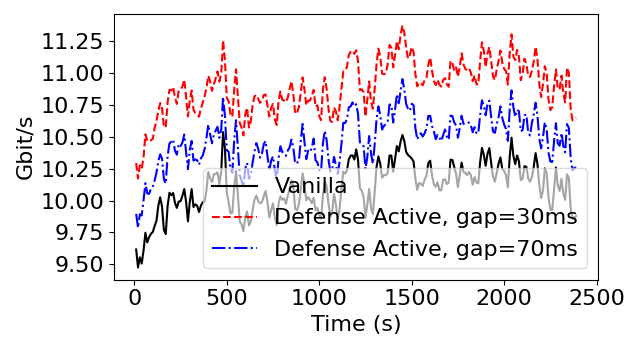}
  \caption{20-seconds moving average of total relay goodput on a 5\% scaled
    down Tor network from April 2021 (340 relays).
    The dropmark defense experiments generate padding in more than 400,000 circuits.}
  \label{fig:dropmark_overhead} 
\end{figure}

We deployed our Dropmark defense machine in a Shadow Tor network
while running the Dropmark detection code on attacker guard relays. The defense is
efficient:
with no attack, the malicious guard
relays incorrectly (i.e. false positives) flagged 99.9975\% of the active circuits which connected to the
simulated Alexa Top list.  

Figure~\ref{fig:dropmark_overhead} shows the overall bandwidth overhead induced
by the Dropmark defense during the simulation. The overhead can be tuned from
the interval length that defines how fast the Dropmark padding machine juggles
between internal states. The intervals [1ms, 30ms] and
[1ms, 70ms] are evaluated, and show a clear connection to the overhead. A
larger interval reduces the bandwidth pressure over the network but could
further delay the CONNECTED cell sent by the exit. Perfect superposition of the
curves comes from Shadow's simulation determinism and from our design
efficiency. Indeed, the graphs would superpose under the following conditions:
\begin{inparaenum} 
  \item Circuit choice is deterministic, leading to the same circuits in the
    different simulations (necessary for a meaningful impact comparison).
  \item The Dropmark defense does not impact user traffic (indeed, the
    padding is sent when the circuit is silent).
  \item There is a surplus of bandwidth in the client-to-middle part of Tor
    circuits due to Tor topologies (allows to absorb the padding overhead).
    Tor has a significant bandwidth excess in these parts of circuits
    for about 10 years now~\cite{tor-metrics, rochet2017waterfilling}.
\end{inparaenum}

During $1$ hour of virtual time,
the $37,617$ emulated Tor clients altogether generated $74,444$ active circuits
every 10 minutes in aggregate. Figure~\ref{fig:dropmark_res} shows that, on average,
with $\textbf{gap}=70ms$ the
middle relays sent $217$ padding cells in these circuits to cover potential
Dropmark attempts,
totaling $\approx 98$ million cells ($\approx 47 GiB$)  during the virtual hour over
the whole network. While these numbers may look large, they represent a
small overhead of $3.6$\% for the overall network goodput (throughput of RELAY command cells).
Moreover, thanks to the deployed protocol restriction,
the defense is only required to cover on average $390$ms of expected protocol
silence (Figure~\ref{fig:dropmark_res}). Varying the defense parameters
allows full control over the overhead. Note that, because recent Tor
topologies have a scarce total exit capacity~\cite{claps-ccs2020}, slightly
increasing the bandwidth usage in the entry and middle parts of Tor circuits is
not expected to impact the clients' goodput with is bottle-necked by exit bandwidth.

\begin{figure}
    \centering
   \begin{minipage}{.49\linewidth}
     \includegraphics[width=\linewidth]{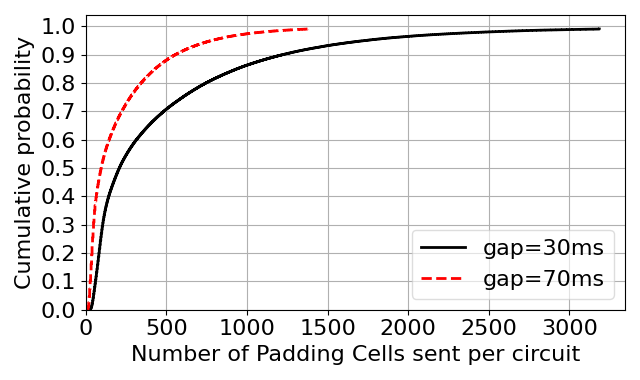}
   \end{minipage}
   \begin{minipage}{.49\linewidth}
     \includegraphics[width=\linewidth]{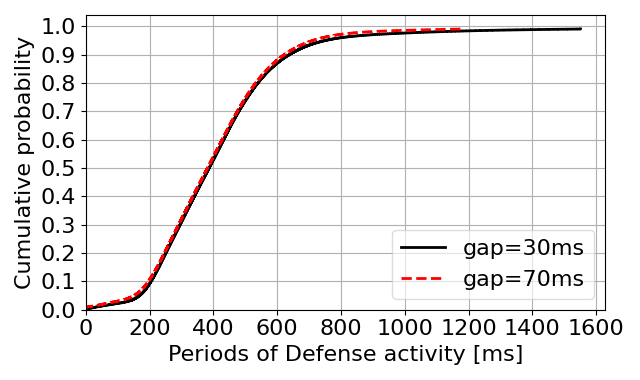}
   \end{minipage}
   \caption{Per-circuit statistics of the dropmark defense deployed on our FAN
     instance, for interval choices [1ms, 30ms] and [1ms, 70ms] in the padding machine
     design. Plotted until percentile .99} 
   \label{fig:dropmark_res}
\end{figure}

\subsubsection{Bayesian Detection Rate}

We assume, per the Dropmark's threat model, that the adversary sits between some fraction of the
clients and the Internet, controlling some guards or ISPs, and listening for Dropmark events.
The challenge for the adversary conducting a traffic confirmation attack is to reliably distinguish whether 
or not an observed Dropmark signal was sent by it.
We
can compute the reliability
using the \textit{Bayesian detection rate} relative to how much bandwidth the
adversary controls at the exit position. We can calculate the Bayesian detection
rate using Bayes' Theorem:
Let $D$ be a random variable that denotes the observation of a Dropmark over the
circuit ($\neg D$ otherwise). Let $F$ be a random variable that denotes the
fraction of exit bandwidth the adversary controls ($\neg F$ the fraction
not controlled). We are interested in
$P(F|D)$, that is, the probability that the observed Dropmark signal actually
comes from the adversary:

$$ P(F|D) = \frac{P(F) \times P(D|F)}{P(F)\times P(D|F) + P(\neg F)\times
  P(D|\neg F)}$$

From our Dropmark empirical results in Table~\ref{table:dropmark} and from the empirical
results of the Dropmark defense obtained with Shadow, we can compute $P(F|D)$ relative to
$F$. Indeed, $P(D|F)=TPos=0.999972$ and $P(D|\neg F)=FPos=0.007713$. Using the
defense, $P(D|\neg F)$ rises to $0.999975$. Figure~\ref{fig:dropmark_bayes}
gives the overall picture and highlights the impact of the \emph{undefended}
Dropmark's low FPR: the attack's reliability rises as the
compromised exit fraction increases. Several facts make these
results concerning: 1) In its recent history, the Tor network has several times experienced a malicious 
operator controlling a significant fraction ($> 20\%$) of the
total exit bandwidth~\cite{nusenu-report}, and 2) several exit families hold a
large fraction of the total exit
bandwidth~\cite{nusenu-ornetstats}. Independent of their
trustworthiness, being in a position that allows reliable
Dropmark-based deanonymization at scale might be undesired until
the problem is solved.

\begin{figure}
  \centering
    \includegraphics[width=0.6\linewidth]{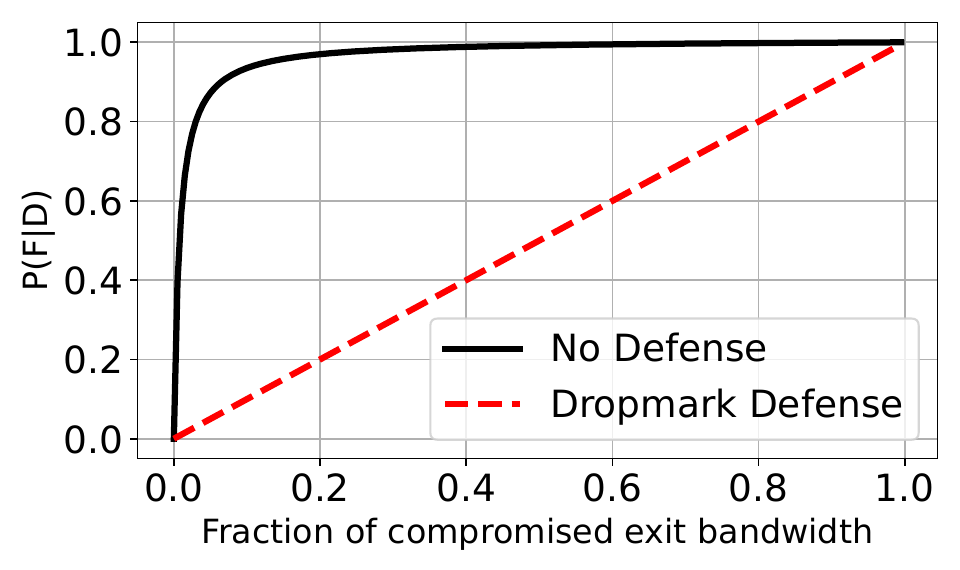}
    \caption{Bayesian Detection Rate as the dropmark's reliability to large
      scale deanonymization (needs to be close to 1 to avoid flagging
      ``innocents'').}
    \label{fig:dropmark_bayes}
\end{figure}

The FAN-based defense is efficient in breaking the Dropmark's reliability, enforcing a
given fingerprint that is independent of the presence of some adversary on the exit
node with a potential adaptive Dropmark strategy.

\subsection{Case Study Conclusion}

This case study illustrates that deploying
ephemeral conservative protocol policies can be useful to defend against traffic
confirmation techniques, reducing the overhead of padding-based defenses. We
argue that such policies would be infeasible to deploy over vanilla Tor due to the impact on
 extensibility, and in case of problems, require choosing between two undesirable options, such as either 
 excluding relays that do not update fast enough or
wait long enough for the restrictive protocol feature to be removed through Tor
volunteers' regular update lifecycle. 

\section{Related Work}
\label{sec:related}

Bento~\cite{bento-sigcomm} and Proteus~\cite{wails-proteus-foci23} are the
closest work to the FAN approach. Like FAN,
Bento is fundamentally an architecture. Bento provides mechanisms for
installing user-defined functions on Tor nodes using a ``middlebox'' 
approach,
assuming a perfectly secure enclave to run user-functions in. Bento pays a high
performance price being designed as a separate process running a Python 
controller
interacting with the Tor process through a socket to exchange information (e.g., Tor cells).
While the two designs are close to each other in spirit, the architectural
differences lead to radical differences in terms of performance,
security, capability, and deployability.  FAN mainly discusses and addresses 
the
need for better extensibility from the authority' perspective, while Bento
argues for benefits of a client-controlled environment remotely running
arbitrary code in secure enclaves on relays.

Proteus~\cite{wails-proteus-foci23} provides a protobuf-like language to
remotely run user-defined protocols on bridge relays. Proteus shares similar
observations and benefits as FAN and Bento, and like Bento, its different
approach makes it an interesting and independent research direction to pursue.

Rochet and Pereira~\cite{dropping-pets2018}
established that protocol flexibility could be leveraged to perform traffic correlation.
We build upon their work and show how a software architecture change can allow us more control
while removing protocol tolerance, enabling a realistic solution against the
Dropmark attack, and open new research perspectives. The Vanguards
addon~\cite{vanguards-announce} produced by the Tor project reacted to the
research in~\cite{dropping-pets2018} by enabling restrictive protocol policies,
tearing down any circuit receiving an unknown cell. We argue that this choice would
force the Tor project to negotiate any new
feature and wait a long time to deploy them (until most of the relays
are upgraded) or force periodic global exclusions of relays, which would damage the
network diversity and throughput.

\section{Limitations}
\label{sec:limitations}

\textit{Technical limitations of current implementation}
From a software security standpoint, JIT engines are currently behind regular compilers in
regard to memory safety mitigations such as Control Flow Integrity (CFI) or
the No Execute bit (NX) for the stored code. These are mitigations to memory
safety bug exploitations. Nonetheless, the steady progress of JIT engines, and the future of
architectures building upon them is promising, both from a performance and security
point of view. Indeed, the JIT
compilation engine can detect the processor on which it compiles some eBPF or
WebAssembly code, and uses its dedicated instructions. It contrasts compiling ahead-of-time
as we do today since distributing packaged binaries requires compiling the code for a generic subset of instructions
to make sure a variety of, say, x86 processors, would understand all
instructions. Usually it means advanced, faster capabilities are unused. From a
security point of view, our design could be extended to
run untrusted code. However, the existence of a vulnerability in the untrusted code does not imply
the existence of an exploitation vector. Technologies such as WebAssembly are
designed to thwart exploitation vectors, making vulnerabilities, if any,
difficult to impossible to exploit compared to a memory safety vulnerability
within a regular binary~\cite{10.1145/3609021.3609306}.

\textit{Fast path overheads}
Currently, accessing host heap-allocated data by a plugin may involve a copy to move the
data into the plugin's memory authorized space, while regular code would pass a
pointer. As an alternative, if some data is meant to eventually be moved into a plugin,
the host could allocate memory within the plugin's memory space since this space
is owned by the host's process which would involve passing a pointer rather than a copy. This would 
require writing and using a wrapper for the OS's memory allocation syscalls for the host.

\textit{Case study's defense limitations}
A malicious client may employ a congestion attack~\cite{sniper14,pointbreak-sec2019}, to build many 
circuits through the same guard relay (but different
middle relays), activate the Dropmark defense and stop reading the client-guard
TCP connection. 
We can deploy several mitigations to avert such an issue. For example, from
Figure~\ref{fig:dropmark_res}, capping the number of padding cells produced by
the dropmark defense to 350  would be effective for $\approx$99\% of Tor circuits.
A legitimate client can switch to a new circuit when it receives this amount of padding cells since the 
circuit is unusable due to the latency or
potentially being flooded by a malicious exit. This strategy could be programmed
within the plugin and instantly deployed, as a further example of what the FAN
architecture enables.

\section{Future work}

FAN could be used to run different protocols atop the same physical network,
multiplexed over the same TCP connection between two relays and challenging
the adversary's assumption about the nature of the traffic. For example,
high-latency, round-based provably secure anonymity could technically co-exist
with a low-latency network on FAN and hide among the crowd of low-latency
users, atop the same TCP connection which the network engine would allow to
share among plugins.

While we discuss a centralized governance model in this paper (i.e., no
third-party plugins, all relays should have the same set of plugins for a given
core version), research in alternative governance models such as enabling
third-parties raises safety \& security concerns that would need addressing.  One major difference
would be that users could enable/propagate in their own circuits a different set of third-party but 
authorized protocol features. While these third-parties may lead more users to embrace and adapt the 
anonymous communication
network to their use case, it opens the door for fingerprinting attacks.
Moreover, the system would have to make sure that different plugins cannot
misbehave together (collude), or monopolize resources. Our current proposal does not
enable third-parties to propagate code, yet this long-term research goal (hence
``Towards'' in the title) motivates FAN's current complexity. Efficient,
low-bandwidth automated and unattended updates as suggested in this paper are
pre-requisite of these capabilities.

\section{Discussion \& Conclusion}
\label{sec:disc_ccl}

In this work, we
argue that
the freedom allowed by the developers in their protocol designs required to 
ensure
the continuous evolution of the network and minimize the burden to roll out 
changes is
also the source of the adversary's strength. 
In Section~\ref{sec:case}, we gave an appealing example
regarding the circuit extension protocol, in which the protocol's tolerance
enabled a chain of exploits. It has led to the most reliable
traffic confirmation technique to date, and put in use by a state agency. We
also argue that the exploitation of  protocol tolerance is the root cause of 
several other reliable attacks against the Tor
design  (DoS using tolerance in the control-flow~\cite{sniper14}, dropmark 
due to
Robustness in the Relay
protocol~\cite{dropping-pets2018}, arbitrary positioning allowed in the HSDir
list~\cite{oakland2013-trawling}, etc.).

Programming the network with a FAN methodology may help prevent those protocol
flaws from existing in the first place. Indeed, with a FAN architecture, we can
naturally tighten the tolerance embedded in feature designs: there are no
forward/backward compatibility concerns. Developers need not be constrained by future, 
as-yet-unknown requirements when programming with FAN. 
Instead, the developers may be as conservative as possible and deploy
new iterations of functionalities with plugins, programming new behavior
as it is needed. FAN applied to Tor has the potential to reduce its protocols'
attack surface, assuming Tor's governance model does not change (e.g., no
third-party plugins).

\section*{Acknowledgements}
This research was partially funded by the CyberExcellence project of the Public Service of Wallonia (SPW Recherche),
convention No.~2110186, and by REPHRAIN under UKRI
grant: EP/V011189/1.

\bibliographystyle{ACM-Reference-Format}
\bibliography{paper}

\appendix
\section{FAN Plugin Transparency Details and Proofs}
\label{app:sec_analysis}

\paragraph*{FAN Plugin Transparency Details \& Proof Sketches}
Newly issued plugins are stored in the FTLs' structured Tree as leaf nodes.
The leaf depends on the plugin's name, which is a concatenation of the FAN
developer's namespace--- 
registered at
the FTL---and the plugin's unique name (e.g.,
fan.project/my\_plugin, where the registered namespace 
``fan.project'' is concatenated with `/' to the plugin's name 
``my\_plugin''). 
This name is input into the hash function $H$ and the output is truncated to the
Tree's depth $D$, resulting in a 
binary representation of the location of the leaf node. For example, in 
Figure~\ref{fig:auth_path}, with $H=sha256$ and $D=3$, the plugin above 
would be stored at leaf $001$. A `0' is a 
left child and `1' is a right child from the Tree root.
Where there are collisions, plugins are stored in an alphabetically-ordered
list  used to calculate the leaf value as follows:

\begin{multline*}
$$leaf=H(name_1||H(plugin_1||meta-info_1);...;\\name_n||H(plugin_n||meta-info_n))$$.
\end{multline*}

Given this, the root of the FTL Tree can be computed, successively concatenating
children and hashing the concatenation, which then enables 
the ability to verify the \textit{Proofs of 
Availability and Absence} with authentication paths (see below).

This Tree offers:
\begin{inparaenum}
\item \textbf{Efficiency.} The proof of availability and the proof of absence (detailed
  below) can be
  computed with space and time complexity in $\Theta(D + \frac{N}{2^D})$ with $N$ the number of
  plugins included and $D$ the Tree depth assuming $H$ emulates a Random Oracle. In a classic Merkle Tree, the
	proof of absence would require monitoring the whole tree. Observe that
	depending on $N$, the FTLs can choose $D$ such that the proofs'
	complexity is close to the lower bound $\Omega(D)$ on average. $D$ must, however, be the same for all FTLs for the next property to hold.
\item \textbf{Unambiguous name-to-plugin mapping \& spurious plugin detection.} A 
  plugin name maps to a single leaf. The FAN developers can detect if 
the FTL maliciously swaps the 
plugin for another from the proof of availability requests to
	each FTL in each epoch. The relays would also detect it from a
	non-matching signed tree root value in the proof of availability (the
value would not be the same as the one broadcasted by the FAN developers to
the network for the current epoch).  
\end{inparaenum}

\begin{figure}
  \centering
  \begin{small}
\begin{tikzpicture}[level/.style={sibling distance=42mm/#1}, level
	distance=6mm] 
	\node (z){$Root = H(h_0 || h_1)$} child
  { node (a) {} child {node
      {} child {node [circle, draw, color=black,fill=green] (b) {} 
    edge from parent [color=black] node[above] {$0$} }
  child {node [draw, circle, inner sep=1.7pt] (d) {$list$} 
    edge from parent [color=red]
	node[above] {$1$}
      }
      edge from parent [color=red]
      node[above] {$0$}
    }
    child {node [circle, draw, color=black, fill=green] {}
      child {node {$leaf$}
         edge from parent
         node[above] {$0$}
      }
      child {node {$leaf$}
        edge from parent
        node[above] {$1$}
      }
      edge from parent [color=black]
      node[above] {$1$}
    }
    edge from parent [color=red]
    node[above] {$0$}
  }
  child {node [circle, draw, color=black, fill=green] {}
    child {node (g){}
      child {node  {$leaf$}
	 edge from parent [color=black]
         node[above] {$0$}
      }
      child {node {$leaf$}
    edge from parent
        node[above] {$1$}
      }
      edge from parent
      node[above] {$0$}
    }
    child {node {}
      child {node {$leaf$}
 	edge from parent
         node[above] {$0$}
      }
      child {node {$leaf$}
	edge from parent [color=black]
        node[above] {$1$}
      }
	edge from parent [color=black]
	node[above] {$1$}
    }
    edge from parent
    node[above] {$1$}
  };
\end{tikzpicture}
\end{small}
\caption{FTL's Name-Structured Merkle Tree List. Plugins are assigned to a leaf
  depending on their name. If multiple plugins are assigned to the same leaf,
  they are arranged in an alphabetically-ordered list by name. The green dots
  are the value provided for the Authentication Path for
  $H(``fan.project/my\_plugin") = 001$\dots and needed to re-compute the root.} 

	\label{fig:auth_path}
\end{figure}

The proof of availability and the proof of absence can both be computed from an
\emph{authentication path} (see an example in Figure~\ref{fig:auth_path}). That is, for a given name, the FTL produces the 
leaf's list and the list of siblings' hashes on the path from the leaf to 
the root. Knowing the expected path in the 
Tree from the name, the verifier can
unambiguously concatenate the siblings and reconstruct the full path
using the list of hash values for a given plugin. 

Then, having the authentication path, the verifier can
re-compute the Root hash value and compare it to the expected one. 
The verifier can infer directly whether the parent's hash value
needs to be computed as $H(path\_hash||sibling\_hash)$ or
$H(sibling\_hash||path\_hash)$ from the plugin's name, since its truncated 
hash value gives the ordering (i.e., the path direction in the Tree).

Relays will load a plugin if and only if:
\begin{inparaenum}
\item The FTL is ``online''. 
\item The plugin push epoch value is lower or equal to the current
	epoch value.
	\item The proof of availability correctly verifies.
	\item The plugin is not marked as ``withdrawn''.
	\item The broadcasted FTL's root hash matches the proof of 
	availability's reconstructed root hash.
\end{inparaenum}

Observe that relays only need to verify the FAN developers' threshold
signature in the event of a dispute. Furthermore, the proof of availability
can be attached to the plugin when it is transmitted to the network. When 
all verifications succeed, the cost of the overall system at runtime
for the non-protesting relays is minimal: no interactions and a few
Hash calculations to verify the proof of availability.

When relay operators protest
against an issued plugin by the protest epoch, they sign
``\$relay\_identity:protest:\$plugin\_name'' and send the signature to at 
least
one FTL, which will broadcast the signature to other
FTLs. Weighing the number of protests, the FAN developers may decide 
to withdraw the
plugin from the network. Once withdrawn, relay operators can request a
proof of absence from all ``online'' FTLs. If the relay operator maliciously
sends multiple 
signatures
(i.e., different ones for the same protest), the FTLs would forward them to the
FAN developers for potential punishment and ignore the protest.

We have the following theorems:

\begin{theorem}
	Under the assumption that $H:\{0, 1 \}^* \rightarrow \{0, 1\}^n$
  behaves as a random oracle, a name maps to a unique authentication path in
the Name-Structured Merkle Tree List with probability $1-\negl$.
	\label{theo:auth_path}
\end{theorem}

\begin{proof}
	
	Let $H = RO : \{0, 1\}^* \rightarrow \{0, 1\}^n$ a random oracle. Let
	$D$ be the depth of the Tree.
	Let
	$\adv$ be an adversary able to manipulate the plugin and the meta-info to
	duplicate an authentication path for the name $name$ assigned to leaf
	$l$:

	$$ l^{\adv} = name||H(plugin^{\adv}||meta-info^{\adv}) $$

	The adversary succeeds if they can choose $plugin^{\adv}$ and
	$meta-info^{\adv}$ such that the parent's value is:
	
		$$H(l||l^{\adv}) = H(l^{\adv}||l)$$
	
	or for any of the $D-1$ parents we have (with $h_0$ and $h_1$ the child
	values of the current level):

		$$H(h_0||h_1) = H(h_1||h_0)$$

	$\adv$ performs $q$ queries to the Random Oracle, each of
	probability $\frac{D}{2^n}$ to succeed. By
	the union bound, the adversary succeeds with probability $\frac{q\times D}{2^n}
	= \negl$. Therefore, the authentication path is unique with probability $1-\negl$.

\end{proof}

\begin{theorem}
	Under the assumption that $H:\{0, 1 \}^* \rightarrow \{0, 1\}^n$ is a
  uniform random function, and given that the authentication path is unique
  (Theorem~\ref{theo:auth_path}), the FAN developers can detect spurious plugins
  with probability $1-\negl$ with algorithmic complexity $\Theta(D +
  \frac{N}{2^D}$) for $D$ the depth of the Name Structured Merkle Tree List and
  $N$ the number of plugins included.
\label{theo:spurious_plug}
\end{theorem}

\begin{proof}
  $\adv$ succeeds to hide $plugin^{\adv}||meta-info^{\adv}$ with probability
  $\negl$ (directly from Theorem~\ref{theo:auth_path}). Hence, the adversary 
  can
  only modify the legitimate leaf, which gets detected at the re-computation of
  the Tree's root using the authentication path (i.e., the re-computed root does
  not match the expected one using the legitimate plugin).

	The recomputation involves $D$ steps to hash each level of the Tree.
Assuming $H$'s output is uniform, the load on each leaf (i.e., the number of
plugins) is expected to be $\frac{N}{2^D}$ on average. Therefore, the FAN
	developers can detect a spurious plugin with probability $1-\negl$  with complexity $\Theta(D + \frac{N}{2^D})$.
\end{proof}

Theorem~\ref{theo:spurious_plug} guarantees that any member of the network 
can
directly catch any attempt to manipulate the proof of availability for a 
given
name. Moreover, if any relay transmits a proof of availability
to a peer relay with a forged bytecode, the peer relay would detect it without
the need to interact with FTLs and without requiring asymmetric cryptography
operations. These facts display the efficiency of our design and the above
theorems offer secure plugin names.

Another important property provided by the FAN Plugin Transparency system is
its non-equivocation. That is, it is not possible for the FTLs to trick one 
party
into believing that a plugin is included and another party to believe the 
plugin is not.

\begin{theorem}
  Under the assumption that $H:\{0,1\}^* \rightarrow \{0,1\}^n$ behaves as a
  random oracle and is collision-resistant, then an FTL admits a unique root
  value for a known depth $D$ and set of included plugins.
  \label{theo:single_root}
\end{theorem}

\begin{proof}
  Let $TR : \{0, 1\}^*_{0} \times \dots \times \{0, 1\}^*_{2^D} \rightarrow \{0,
  1\}^n$ a function that takes in input the Tree leaves and computes the Tree
  root with successive applications of $H$. At each $H$ application, a collision
  with the root value may happen with probability $\frac{1}{2^n}$. By the union
  bound, the probability that the root is unique is given by $1-\frac{D}{2^n} =
  1 - \negl$.
\end{proof}

From Theorem~\ref{theo:single_root} and assuming a broadcast channel to exchange
FTL's root value (e.g., the Tor Consensus Document), then an equivocation would
lead to an observable root mismatch.

\section{Design Implementation}
\label{app:design_impl}
Figure~\ref{fig:pluginfile} gives an example of meta-info within our FAN implementation. The first line is 
the
sandboxed memory the plugin requires, which gives the right to all code in the
plugin to allocate, share and access data within a specific memory space of $N$
bytes to be specified in the plugin meta-information and authorized by the
plugin manager. Each of the subsequent lines attaches an \textit{entry point}
to an internal hook within the Tor binary.

\section{Simple Plugin Code Example}
\label{cf-code}

An example of a \emph{hook} in C-Tor written in the Relay protocol to intercept
any unknown Relay message is given on Figure~\ref{fig:hook}.

\begin{figure}[H]
    \footnotesize
\begin{minted}{C}
entry_point_map_t pmap;
memset(&pmap, 0, sizeof(pmap));
pmap.ptype = PLUGIN_DEV;
pmap.putype = PLUGIN_CODE_ADD;
pmap.pfamily = PLUGIN_PROTOCOL_RELAY;
pmap.entry_name = (char*)"relay_process_edge_unknown";
caller_id_t caller = RELAY_PROCESS_EDGE_UNKNOWN;
relay_process_edge_t args;
args.circ = circ;
args.layer_hint = layer_hint;
args.edgeconn = conn;
args.cell = cell;

if (invoke_plugin_operation_or_default(&pmap,
    caller, (void*)&args)) {
  log_fn(LOG_PROTOCOL_WARN, LD_PROTOCOL,
      "Received unknown relay command \%d. But"
      " we do not have a developer plugin"
      " able to handle it, destroy the circuit",
      rh->command);
  return -1;
}
return 0;
\end{minted}
\caption{Hook located in Tor's Relay protocol to intercept any unknown Relay
  cell, and destroy the circuit if a plugin does not exist to handle it.
  Originally, the unknown message is ignored as per the Robustness principle.}
\label{fig:hook}
\end{figure}

Furthermore, Figure~\ref{fig:code_example} shows an example of a function
(rather than a message) that we may intercept with a plugin, assuming the
original function call is defined as the default code of a hook. On the right
the original code, on the left the same code written in the C subset compiling
to eBPF.

\begin{figure*}[t]
  \footnotesize
  \begin{minted}[mathescape,
               linenos,
               numbersep=5pt,
               frame=lines,
               framesep=2mm]{text}
memory 16777216
circpad_global_machine_init protocol_circpad add circpad_dropmark_def.o
circpad_setup_machine_on_circ_add protocol_circpad add circpad_dropmark_circ_setup.o
relay_process_edge_unknown protocol_relay add circpad_dropmark_receive_sig.o
connedge_connection_ap_handshake_send_begin_add protocol_conn_edge param 1 add circpad_dropmark_send_sig.o
connedge_received_connected_cell_add protocol_conn_edge param 2 add circpad_dropmark_send_sig.o
circpad_send_padding_cell_for_callback_replace protocol_circpad replace circpad_dropmark_send_padding_cell.o
  \end{minted}
  \caption{Example of a .plugin file containing meta-information for the plugin
	to compile and link with Tor's code. The first line is the maximum heap memory that the plugin can allocate (in bytes; $16$ MiB here). The following lines announce one entry point each. The first element is the Hook name. The second element is the protocol name, the third element is an operation over the hook (e.g., add or replace). There is an optional parameter argument and the final element is the bytecode filename.}
  \label{fig:pluginfile}
\end{figure*}

\begin{figure*}
  \begin{minipage}[t]{.58\textwidth}
    \footnotesize
\begin{minted}[breaklines]{c}
#include "core/or/or.h"
#include "core/or/relay.h"
#include "core/or/plugin.h"
#include "core/or/plugin_helper.h"

uint64_t
sendme_circuit_data_received(relay_process_edge_t *pedge) {
  int deliver_window, domain;
  circuit_t *circ = (circuit_t *) get(RELAY_ARG_CIRCUIT_T, 1, pedge);
  crypt_path_t *layer_hint = (crypt_path_t *) get(RELAY_ARG_CRYPT_PATH_T, 1, pedge);

  if ((int) get(UTIL_CIRCUIT_IS_ORIGIN, 1, circ)) {
    deliver_window = (int) get(RELAY_LAYER_HINT_DELIVER_WINDOW, 1, layer_hint);
    set(RELAY_LAYER_HINT_DELIVER_WINDOW, 2, layer_hint, --deliver_window);
    domain = LD_APP;
  } else {
    deliver_window = (int) get(RELAY_CIRC_DELIVER_WINDOW, 1, circ);
    set(RELAY_CIRC_DELIVER_WINDOW, 1, circ, --deliver_window);
    domain = LD_EXIT;
  }

  log_fn_(domain, LD_PLUGIN, __FUNCTION__,
    "Circuit deliver_window now %d.", deliver_window);
  return 0;
}
\end{minted}
\end{minipage}
  \hfill
\begin{minipage}[t]{.38\textwidth}
  \footnotesize
\begin{minted}[breaklines]{c}





int
sendme_circuit_data_received(circuit_t *circ, crypt_path_t *layer_hint) {
  int deliver_window, domain;

  if (CIRCUIT_IS_ORIGIN(circ)) {
    tor_assert(layer_hint);
    --layer_hint->deliver_window;
    deliver_window = layer_hint->deliver_window;
    domain = LD_APP;
  } else {
    tor_assert(!layer_hint);
    --circ->deliver_window;
    deliver_window = circ->deliver_window;
    domain = LD_EXIT;
  }

  log_debug(domain,
    "Circuit deliver_window now %d.", deliver_window);
  return deliver_window;
}
\end{minted}
\end{minipage}
\caption{Code example of the smallest function we intercept in one of our plugins
  to evaluate the CPU overhead. The left function is defined in a separate .c
  file from the main Tor source code and compiled to a BPF bytecode object. It
  uses a get/set API to indirectly interact with values outside its
  allocated memory, upon the control of the main binary. The right function is
  the one existing in Tor's source code.} \label{fig:code_example}
\end{figure*}

\end{document}